\documentclass[12pt]{article}
\usepackage[utf8]{inputenc}
\usepackage[british]{babel}
\usepackage{cmap}
\usepackage{lmodern}

\usepackage{amssymb, amsmath, amsthm}
\usepackage[a4paper,top=25mm,bottom=25mm,left=25mm,right=25mm]{geometry}
\usepackage{ragged2e}

\usepackage{authblk} 
\usepackage{pifont}
\usepackage{graphicx}
\usepackage[dvipsnames,svgnames,table]{xcolor}
\usepackage[figuresright]{rotating}
\usepackage{xtab} 
\usepackage{longtable} 
\usepackage{multirow}
\usepackage{footnote}
\usepackage[stable]{footmisc}
\usepackage{chngpage} 
\usepackage{pdflscape} 
\usepackage[nottoc,notlot,notlof]{tocbibind} 

\usepackage{pgfplots}
\pgfplotsset{every tick label/.append style={font=\footnotesize}}
\pgfplotsset{compat=1.18}
\usepackage{setspace}

\usepackage{array}
\newcolumntype{K}[1]{>{\centering\arraybackslash$}p{#1}<{$}}

\makesavenoteenv{tabular}
\usepackage{tabularx}
\usepackage{booktabs}
\usepackage{threeparttable}
\usepackage[referable]{threeparttablex} 
\newcolumntype{R}{>{\raggedleft\arraybackslash}X}
\newcolumntype{L}{>{\raggedright\arraybackslash}X}
\newcolumntype{C}{>{\centering\arraybackslash}X}
\newcolumntype{A}{>{\columncolor{gray!25}}C}
\newcolumntype{a}{>{\columncolor{gray!25}}c}

\newlength{\tablen}

\usepackage{dcolumn} 
\newcolumntype{.}{D{.}{.}{-1}}

\usepackage{tikz}
\usetikzlibrary{arrows, calc, matrix, patterns, positioning, trees}
\usepackage[semicolon]{natbib} 
\usepackage[hyphens]{xurl}
\usepackage{microtype}
\usepackage[justification=centering]{caption} 

\usepackage[labelformat=simple]{subcaption}

\DeclareCaptionLabelFormat{parenthesis}{(#2)}
\makeatletter
\renewcommand\p@subfigure{\arabic{figure}.}
\makeatother

\DeclareCaptionLabelFormat{parenthesis}{(#2)}
\makeatletter
\renewcommand\p@subtable{\arabic{table}.}
\makeatother

\usepackage{enumitem}

\setlist[itemize]{leftmargin=2.5\parindent}
\setlist[enumerate]{leftmargin=2.5\parindent}

\usepackage{hyperref} 
\hypersetup{
  colorlinks   = true,    		
  urlcolor     = blue,    		
  linkcolor    = blue,    		
  citecolor    = ForestGreen	
}

\def\addlegendimage{\csname pgfplots@addlegendimage\endcsname}

\theoremstyle{plain}

\newtheorem{proposition}{Proposition}

\theoremstyle{definition}

\newtheorem{definition}{Definition}

\theoremstyle{remark}

\makeatletter
\let\@fnsymbol\@alph
\makeatother

\def\keywords{\vspace{.5em} 
{\noindent \textit{Keywords}: }}

\def\AMS{\vspace{.5em} 
{\noindent \textbf{\emph{MSC} class}: }}

\def\JEL{\vspace{.5em} 
{\noindent \textbf{\emph{JEL} classification number}: }}

\title{Misaligned incentives in sports: A mathematical analysis of the post-2024 UEFA Champions League qualification}

\author{\href{https://sites.google.com/view/laszlocsato}{L\'aszl\'o Csat\'o}\thanks{~Corresponding author. Email: \emph{laszlo.csato@sztaki.hun-ren.hu} \newline
Institute for Computer Science and Control (SZTAKI), Hungarian Research Network (HUN-REN), Laboratory on Engineering and Management Intelligence, Research Group of Operations Research and Decision Systems, Budapest, Hungary \newline
Corvinus University of Budapest (BCE), Institute of Operations and Decision Sciences, Department of Operations Research and Actuarial Sciences, Budapest, Hungary}
$\qquad \qquad$
Sergey Ilyin\thanks{~Email: \emph{ilyin-sergey@yandex.com} \newline
Independent researcher, T\'orshavn, Faroe Islands}
}

\date{\today}

\def\Dedication{
{\noindent
\emph{Hoc est, quod ait Heraclitus: ``In idem flumen bis descendimus et non descendimus.''}\footnote{~Source: \url{https://www.perseus.tufts.edu/hopper/text?doc=Perseus\%3Atext\%3A2007.01.0080\%3Aletter\%3D58}}
}
\flushright
\noindent (Seneca: \emph{Epistulae Morales ad Lucilium LVIII})
\vspace{0.5cm} 
\justify }

\begin{document}


\maketitle
\thispagestyle{empty}
\Dedication

\begin{abstract}
\noindent
UEFA declares that it is committed to respecting the fundamental values of sports. However, the qualification rules of the post-2024 UEFA Champions League are shown to be unfair: a game with misaligned incentives was narrowly avoided in the 2023/24 German Bundesliga. We develop a mathematical model to reveal how incentives for losing can be reduced or eliminated. Since UEFA repeatedly commits the same theoretical mistake in designing the qualification system of its competitions, governing bodies in sports are called to work more closely together with the scientific community.


\AMS{62F07, 90B90, 91B14}

\JEL{C44, D71, Z20}
\end{abstract}

\clearpage
\restoregeometry

\section{Introduction} \label{Sec1}

Ensuring the fairness and integrity of sports is a crucial responsibility of tournament organisers \citep{ForrestMcHale2019, KendallLenten2017, DevriesereCsatoGoossens2024}. One of the most important criteria of fairness is incentive compatibility: it should never be allowed to happen that one player or team seeks to win and its opponent to lose. This would be clearly against the spirit of the game and may easily lead to collusion or match-fixing \citep{Csato2024f, Haigh2009}.

Naturally, the Union of European Football Associations (UEFA) is also committed to these values, at least in declarations and media releases. For example, UEFA president \emph{Aleksander {\v C}eferin} has commented on the novel format of the UEFA Champions League in 2024 with the following words \citep{UEFA2024c}:
``\emph{UEFA has clearly shown that we are fully committed to respecting the fundamental values of sport and to defending the key principle of open competitions, with qualification based on sporting merit, fully in line with the values and solidarity-based European sports model.}''

However, is the UEFA Champions League qualification based on sporting merit under all circumstances? Is there any issue in its rules that threatens the fairness of the tournament and can potentially jeopardise competition?

The current paper uncovers that the assignment of the European performance spots, introduced from the 2024/25 season of the UEFA Champions League, may create a misaligned incentive to lose. This was narrowly avoided in the last round of the 2023/24 German Bundesliga, which determined qualification for the next season of the UEFA Champions League. To analyse the issue in depth, we develop a mathematical model to propose two solutions that can improve fairness.

The novel case of incentive incompatibility is especially worrying because ``\emph{the objective for a sport organisation should be to eliminate any game that raises potential unfairness ex-ante}'' \citep[p.~535]{ScellesFrancoisValenti2024}.
Furthermore, the same theoretical problem led to some instances of incentive incompatibility in the past, even though several academic papers have been published on this issue (see Section~\ref{Sec2} and Table~\ref{Table2}). In order to avoid these paradoxes, governing bodies in major sports, especially UEFA, are called to strengthen collaboration with the scientific community.


\section{Related literature} \label{Sec2}

Tanking---the act of deliberately losing a game---is against the spirit of sports and is a well-known violation of fairness in sports \citep{KendallLenten2017, Csato2021a}. Tanking is closely related to the axiom of incentive compatibility, which requires that tournament rules should encourage honest performance under any circumstances; otherwise, a team could face misaligned or perverse incentives that might inspire tanking.

Probably the most serious type of incentive \emph{in}compatibility is when a team can be strictly better off by losing. According to our knowledge, such a possibility has been identified first by \citet{DagaevSonin2018} in the UEFA Europa League. Several qualification tournaments for the FIFA World Cups and the UEFA European Championships have been verified to suffer from the same weakness \citep{Csato2020c}, even allowing for a match where the dominant strategy of both opponents is to play a draw \citep{Csato2020d}.

In the past, UEFA club competitions violated strategy-proofness mainly due to the unfair allocation of vacant slots. Vacant slots occur if teams can qualify for a tournament in more than one way, thus, the organiser needs to specify what happens if two conditions of qualification are simultaneously satisfied. A straightforward example is the UEFA Europa League, where teams can enter via both their domestic championship and domestic cup. According to \citet{DagaevSonin2018}, these tournament systems consisting of multiple round-robin and knockout tournaments are incentive incompatible if
(1) there are at least two round-robin tournaments with at least one prize;
(2) there is one round-robin tournament and at least one knockout tournament with at least one prize, and the vacant slots are not always filled through the round-robin tournament.

This theoretical result might not seem especially interesting for sports policy. However, the misaligned allocation rule has been responsible for the incentive incompatibility of
\begin{itemize}
\item
the UEFA Europa League entry until the 2014/15 season \citep{DagaevSonin2018};
\item
the UEFA Champions League entry between the 2015/16 and 2017/18 seasons \citep{Csato2019d};
\item
the UEFA Champions League seeding from the 2015/16 season \citep{Csato2020a}.
\end{itemize}
Fortunately, the third problem has been solved by the new seeding rules effective from the 2024/25 season, see Section~\ref{Sec5}. On the other hand, the reform has also created potentially misaligned incentives in the UEFA Champions League qualification as will be shown in Section~\ref{Sec3}.

The introduction of the UEFA Nations League has also allowed for matches where a team has an incentive to lose \citep{Csato2022a, HaugenKrumer2021}. Even though \citet{ScellesFrancoisValenti2024} identify no such games between 2018 and 2021, the authors call for more research around optimal competition formats that provide appropriate incentives for \emph{all} teams in \emph{all} games. Analogously, a recent overview of the causes and consequences of fraud in sport \citep{VanwerschWillemConstandtHardyns2022} concludes that future research is necessary to help the competent authorities in tackling fraud and to gather more knowledge on actionable elements with regard to its prevention.
Our paper can be considered as one of these follow-up studies.

\section{A real-world example} \label{Sec3}

The UEFA Champions League is the most prestigious European club football tournament. The number of participating teams for each national association is determined by the access list, which ranks the associations based on their UEFA association coefficients \citep{Csato2022b}.
The clubs can enter the main league phase of the UEFA Champions League in several ways:
\begin{itemize}
\item
The highest-ranked associations have a given number of slots in the league phase. For instance, the four highest-ranked associations---including Germany---are currently represented by four clubs.
\item
Teams from lower-ranked associations play in the UEFA Champions League qualification. The qualifying consists of two independent knockout competitions; the champions path for the champions of the lower-ranked associations and the league path for the teams of some higher-ranked associations in order to guarantee that a given number of champions play in the league phase.
\item
Both the UEFA Champions League and UEFA Europa League (the second most prestigious European club football competition) titleholders automatically qualify for the league phase, even if they do not qualify for the league phase through their domestic championship.
\item
The two associations with the most coefficient points in the previous season have one European performance spot each in the league phase.
\end{itemize}

However, if a titleholder qualifies for the league phase via its domestic championship, then a vacancy is created in the league phase, which is filled according to \citet[Article~3.04]{UEFA2024b}:

\begin{enumerate}[label=\alph*.]
\item
\emph{A vacancy created by the UEFA Champions League titleholder is filled by the domestic champion with the highest individual club coefficient of all the clubs that qualify for the champions path.}
\item
\emph{A vacancy created by the UEFA Europa League titleholder is filled by the club with the best individual club coefficient of all the clubs that qualify for the champions path and the league path, provided that this club is the highest ranked domestically of those from its association that have not already qualified for the league phase of the competition directly. If this condition of being the highest ranked domestic club is not met by the club with the highest individual club coefficient, then the position moves to the club with the next highest individual club coefficient of all clubs in the champions path and league path.}
\end{enumerate}

The so-called European performance spot has been introduced from the 2024/25 season \citep[Article~3.07]{UEFA2024b}: \\
\emph{Finally, the two associations whose affiliated clubs achieved the best collective performance in the previous season of the UEFA men's club competitions (i.e.\ best season club association coefficient), in accordance with the season coefficient principles (see Annex D.3), are each entitled to an additional place in the league phase, known as a European performance spot (EPS). The EPS is allocated to the club that finishes the relevant association's domestic championship in the highest position of all those that do not qualify for the league phase of the UEFA Champions League via the domestic championship (after any vacancies have been filled in accordance with Articles 3.04 to 3.06). The two associations' overall quota of clubs qualifying to UEFA club competitions are each increased by one club.}

Since the clubs qualify for the UEFA Champions League based on the results of the previous season, the allocation of the European performance spot could have lead to misaligned incentives in the 2023/24 season of domestic leagues.

\begin{table}[t!]
\caption{(Hypothetical) standing of the 2023/24 \\ German Bundesliga before the last matchday}
\label{Table1}
\begin{threeparttable}
\rowcolors{3}{}{gray!20}
    \begin{tabularx}{\linewidth}{Cl CCC CCC >{\bfseries}C} \toprule \hiderowcolors
    Pos   & Team  & W     & D     & L     & GF    & GA    & GD    & Pts \\ \bottomrule \showrowcolors
    1     & Bayer Leverkusen & 27    & 6     & 0     & 87    & 23    & $+$64    & 87 \\
    2     & Bayern M\"unchen & 23    & 3     & 7     & 92    & 41    & $+$51    & 72 \\
    3     & VfB Stuttgart & 22    & 4     & 7     & 74    & 39    & $+$35    & 70 \\
    4     & RB Leipzig & 19    & 7     & 7     & 75    & 37    & $+$38    & 64 \\ \hline
    5     & Borussia Dortmund & 18    & 9     & 6     & 66    & 43    & $+$23    & 63 \\ \hline
    6     & Eintracht Frankfurt & 11    & 13    & 9     & 49    & 48    & $+$1     & 46 \\ \hline
    7     & SC Freiburg & 11    & 9     & 13    & 44    & 56    & $-$12   & 42 \\
    8     & 1899 Hoffenheim & 11    & 7     & 15    & 62    & 66    & $-$4    & 40 \\ \bottomrule    
    \end{tabularx}    
    \begin{tablenotes} \footnotesize
\item
The result of the match Borussia Dortmund vs.\ 1899 Hoffenheim has been changed from 2-3 to 4-3.
\item
The first four teams automatically qualify for the UEFA Champions League due to the access list.
\item
The fifth team also qualifies due to the European performance spot. However, if Dortmund is ranked fifth and wins the Champions League final, this slot will be obtained by the sixth-ranked team.
\item
Pos = Position; W = Won; D = Drawn; L = Lost; GF = Goals for; GA = Goals against; GD = Goal difference; Pts = Points. All teams have played 33 matches. 
    \end{tablenotes}
\end{threeparttable}
\end{table}

Table~\ref{Table1} shows the standing of the 2023/24 German Bundesliga on the morning of 18 May 2024 (the last matchday) with a slight modification. In particular, the result of the match Borussia Dortmund vs.\ 1899 Hoffenheim, played on 15 February 2024, is changed from 2-3 to 4-3. At that moment, it was already known that Germany and Italy finished in the top two based on association club coefficients in the 2023/24 season, and earned the two European performance spots in the 2024/25 Champions League. The two clubs playing in the final of the UEFA Champions League on 1 June 2024 were known to be Borussia Dortmund and the Spanish club Real Madrid, too.

Consider the situation of Eintracht Frankfurt, which plays against RB Leipzig at home on the last (34th) matchday. It could be ranked neither higher nor lower than its current sixth position.
Since the top four teams automatically qualify for the UEFA Champions League, the uncertainty with respect to this tournament resides in whether Borussia Dortmund or RB Leipzig obtains the fourth position, and whether Borussia Dortmund wins the 2023/24 UEFA Champions League or not.

The fate of Eintracht Frankfurt is determined as follows:
\begin{itemize}
\item
If Borussia Dortmund is the fourth and
\begin{itemize}[label=$\diamond$]
\item
Dortmund wins the 2023/24 UEFA Champions League final, then a vacancy is created in the league phase, which is filled by the Ukrainian champion Shakhtar Donetsk according to \citet[Article~3.04]{UEFA2024b};
\item
Dortmund loses the final, then a vacancy may be created by the UEFA Champions League titleholder Real Madrid, which would be filled by the Ukrainian champion Shakhtar Donetsk according to \citet[Article~3.04]{UEFA2024b}.
\end{itemize}
Consequently, the European performance spot obtained by the German Bundesliga is given to the fifth-ranked team RB Leipzig according to \citet[Article~3.07]{UEFA2024b}. Thus, the first five teams qualify for the UEFA Champions League independently of which team wins the 2023/24 UEFA Champions League final.
\item
If Borussia Dortmund is the fifth and
\begin{itemize}[label=$\diamond$]
\item
Dortmund wins the 2023/24 UEFA Champions League final, then it qualifies for the next season of this tournament as the titleholder, and the European performance spot is given to the sixth-ranked Eintracht Frankfurt according to \citet[Article~3.07]{UEFA2024b};
\item
Dortmund loses the final, then Dortmund obtains the European performance spot of the German Bundesliga according to \citet[Article~3.07]{UEFA2024b}, while a vacancy may be created by the UEFA Champions League titleholder Real Madrid, which would be filled by the Ukrainian champion Shakhtar Donetsk according to \citet[Article~3.04]{UEFA2024b}.
\end{itemize}
\end{itemize}
In the first case when Dortmund is ranked fourth, Eintracht Frankfurt could never play in the UEFA Champions League. However, in the second case when Dortmund is ranked fifth, Eintracht Frankfurt can enter the Champions League with a non-negligible probability (if Dortmund wins the 2023/24 UEFA Champions League).
Hence, the dominant strategy of Eintracht Frankfurt is to guarantee that Borussia Dortmund is ranked fifth in the 2023/24 German Bundesliga, which can be easily achieved by losing against RB Leipzig on 18 May 2024.

\begin{figure}[t]
\centering

\begin{tikzpicture}[scale=1,auto=center, transform shape, >=triangle 45]
\tikzstyle{every node}=[draw,align=center];
  \node (C1) at (0,6) {Is the team ranked between 1st and 4th?};
  \node (C2) at (3,3) {Does \emph{Borussia Dortmund} \\ win the 2023/24 UEFA \\ Champions League final?};
  \node (C3) at (-3,0) {Is the 5th-ranked team \\ \emph{Borussia Dortmund}?};
  \node (C4) at (7,0) {Is the team ranked 5th?};
  \node (C5) at (-3,-3) {Is the team ranked 5th?};
  \node (C6) at (3,-6) {Is the team ranked 6th?};

\tikzstyle{every node}=[align=center];
  \node (O1) at (-3,3) {\textcolor{blue}{Qualification (access list)}};
  \node (O2) at (3,-3) {\textcolor{blue}{Qualification (EPS)}};
  \node (O3) at (-3,-6) {\textcolor{blue}{Qualification (Titleholder)}};
  \node (O4) at (8,-6) {\textcolor{red}{Elimination}};
  
\tikzstyle{every node}=[align=center];  
  \draw [->,line width=1pt] (C1) -- (O1)  node [midway,above left] {Yes};
  \draw [->,line width=1pt] (C1) -- (C2)  node [midway,above right] {No};  
  \draw [->,line width=1pt] (C2) -- (C3)  node [midway,above left] {Yes};
  \draw [->,line width=1pt] (C2) -- (C4)  node [midway,above right] {No};
  \draw [->,line width=1pt] (C3) -- (C5)  node [midway,left] {Yes};
  \draw [->,line width=1pt] (C4) -- (O2)  node [midway,above left] {Yes};
  \draw [->,line width=1pt] (C4) -- (O4)  node [midway,above right] {No};
  \draw [->,line width=1pt] (C5) -- (C6)  node [midway,above right] {No};
  \draw [->,line width=1pt] (C5) -- (O3)  node [midway,left] {Yes};
  \draw [->,line width=1pt] (C6) -- (O2)  node [midway,left] {Yes};
  \draw [->,line width=1pt] (C6) -- (O4)  node [midway,above] {No};
\end{tikzpicture}

\captionsetup{justification=centering}
\caption{Qualification for the 2024/25 UEFA Champions League \\ from the 2023/24 German Bundesliga}
\label{Fig1}
\end{figure}


Figure~\ref{Fig1} provides a graphical illustration of how the 2024/25 Champions League spots are allocated for the German teams. Obviously, the only way for the sixth-ranked team to qualify is if Borussia Dortmund is ranked fifth and wins the 2023/24 UEFA Champions League final simultaneously.

Note that the main problem resides in the strange interaction between \citet[Article~3.04]{UEFA2024b} and \citet[Article~3.07]{UEFA2024b}. In particular, the order in which they are applied depends on the ranking in the 2023/24 German Bundesliga that can be exploited by a tanking strategy outlined in \citet{DagaevSonin2018}: in a round-robin tournament, a situation always exists in which a deliberate loss does not affect the rank of a particular team (in our case, Eintracht Frankfurt) but changes the order of two higher-ranked teams (in our case, RB Leipzig and Borussia Dortmund).

The above pathological situation has not been realised as Borussia Dortmund has been guaranteed to finish in the fifth place before the last round. Nonetheless, UEFA has clearly had good luck in avoiding a strange game in the 2023/24 German Bundesliga.
Some football fans have also noticed the problem of misaligned incentives for Eintracht Frankfurt, see \url{https://m.sports.ru/tribuna/blogs/gazzzzzpacho/3239387.html}.

\section{A mathematical model} \label{Sec4}

This section presents a formal analysis of the problem.
However, there is no need to understand the mathematical framework for practitioners in sports organizations or policymakers who are mainly interested in the implications since all these issues will be thoroughly discussed in Section~\ref{Sec5}.

A \emph{tournament} $(X,Y)$ is a tuple of a \emph{set of results} $X$ and a \emph{ranking method} $Y$ that assigns a linear order of the teams to any set of results $X$. The number of teams in tournament $(X,Y)$ is denoted by $\lvert X \rvert$. For each team $1 \leq i \leq \lvert X \rvert$, $Y_i(X)$ is its rank under a given set of results $X$.

Let $\mathbf{W} = \left[ (S,P), (T,Q), (U,R) \right]$ be a \emph{tournament system} consisting of three tournaments $(S,P)$, $(T,Q)$, $(U,R)$.
The winner of tournament $(S,P)$ is denoted by $p$, that is, $1 \leq p \leq \lvert S \rvert$ is the (unique) team for which $P_p(S) = 1$ holds. This team $p$ also plays in tournament $(T,Q)$ but not in tournament $(U,R)$.
The winner of tournament $(U,R)$ is denoted by $r$: $1 \leq r \leq \lvert U \rvert$ is the (unique) team for which $R_r(U) = 1$ holds.

This model can represent the qualification of the UEFA Champions League as follows. $(S,P)$ is the Champions League. $(T,Q)$ is the domestic league of the Champions League winner, denoted by $p$. $(U,R)$ is the ``competition'' played by all teams qualified for the champions path in the qualification of the Champions League such that the ranking method $R$ is based on the club coefficients of these teams.

The \emph{allocation rule} $\mathcal{A} (\mathbf{W})$ chooses the winner of tournament $(S,P)$ and the top $q$ teams of tournament $(T,Q)$.
However, a \emph{vacancy} may be created if team $p$ is ranked among the top $q$ teams in tournament $(T,Q)$.

We consider three allocation rules that differ in the assignment of the vacant slot caused by team $p$.
As discussed in Section~\ref{Sec3}, the current rule of the UEFA Champions League is:
\[
\mathcal{A} (\mathbf{W}) = 
\begin{cases}
    i: i = p \lor Q_i(T) \leq q & \text{if } Q_p(T) \geq q+1 \\
    i: i = r \lor Q_i(T) \leq q & \text{if } Q_p(T) \leq q-1 \\
    i: Q_i(T) \leq q+1      & \text{if } Q_p(T) = q.
\end{cases}
\]
Note that the winner $r$ of tournament $(U,R)$ receives the vacant slot only if the titleholder team $p$ is ranked among the top $q-1$ teams in tournament $(T,Q)$ because the last, $q$th slot of tournament $(T,Q)$ is the European performance spot.

The allocation rule can be simplified if the order of Articles~3.04 and 3.07 in the UEFA Champions League regulation \citep{UEFA2024b} is reversed. Then the two European performance spots are allocated first \citep[Article~3.07]{UEFA2024b}, followed by addressing the possible vacancies created by the titleholders \citep[Article~3.04]{UEFA2024b}, which leads to:
\[
\mathcal{A}' (\mathbf{W}) = 
\begin{cases}
    i: i = p \lor Q_i(T) \leq q & \text{if } Q_p(T) \geq q+1 \\
    i: i = r \lor Q_i(T) \leq q & \text{if } Q_p(T) \leq q.
\end{cases}
\]

Finally, any vacancy created by a titleholder can be filled by the highest-ranked not already qualified club from its association:
\[
\mathcal{A}'' (\mathbf{W}) = 
\begin{cases}
    i: i = p \lor Q_i(T) \leq q & \text{if } Q_p(T) \geq q+1 \\
    i: Q_i(T) \leq q+1      & \text{if } Q_p(T) \leq q.
\end{cases}
\]

Now the incentive (in)compatibility of these allocation rules will be explored.

\begin{definition} \label{Def1}
Tournament $(S,P)$ is called \emph{manipulable} by team $1 \leq i \leq \lvert S \rvert$ if there exists a tournament $\left( S',P \right)$ such that $S' = S$ except for a worse performance by team $i$, and the following two conditions hold:
\begin{itemize}
\item
$P_i \left( S \right) = P_i \left( S' \right)$;
\item
$P_j \left( S \right) \geq 2$ but $P_j \left( S' \right) = 1$ for a team $1 \leq j \leq \lvert S \rvert$.
\end{itemize}
\end{definition}
That is, a tanking strategy by team $i$ guarantees the same final ranking for team $i$ and makes team $j$ the tournament winner.

\begin{definition} \label{Def2}
Tournament $(T,Q)$ is called \emph{manipulable} by team $1 \leq i \leq \lvert T \rvert$ if there exists a tournament $\left( T',Q \right)$ such that $T' = T$ except for a worse performance by team $i$, and the following two conditions hold:
\begin{itemize}
\item
$Q_i \left( T \right) = Q_i \left( T' \right) = q+1$;
\item
$Q_j \left( T \right) = Q_k \left( T' \right) = q-1$, as well as $Q_j \left( T' \right) = Q_k \left( T \right) = q$ for some teams $1 \leq j,k \leq \lvert T \rvert$.
\end{itemize}
\end{definition}
That is, a tanking strategy by team $i$ guarantees that it remains ranked $(q+1)$th, while reverses teams $j$ and $k$ in the two positions directly above ($(q-1)$th and $q$th).

In tournament $(U,R)$, no tanking strategy exists because the UEFA club coefficient cannot be increased by a worse performance on the field.

\begin{proposition} \label{Prop1}
Tournament system $\mathbf{W} = \left[ (S,P), (T,Q), (U,R) \right]$ with the allocation rule $\mathcal{A} (\mathbf{W})$ is incentive incompatible if and only if one of the following conditions hold:
\begin{itemize}
\item
tournament $(S,P)$ is manipulable by the team $1 \leq r \leq \lvert U \rvert$ (for which $R_r(U) = 1$) such that the team $j$ in Definition~\ref{Def1} plays in both tournaments $(S,P)$ and $(T,Q)$ with $Q_j(T) \leq q-1$; or
\item
tournament $(S,P)$ is manipulable by the team $1 \leq i \leq \lvert T \rvert$ for which $Q_i(T) = q+1$ such that the team $j$ in Definition~\ref{Def1} plays in both tournaments $(S,P)$ and $(T,Q)$ with $Q_j(T) = q$; or
\item
tournament $(T,Q)$ is manipulable by the team $1 \leq i \leq \lvert T \rvert$ for which $Q_i(T) = q+1$ such that the team $j$ in Definition~\ref{Def2} plays in both tournaments $(S,P)$ and $(T,Q)$ with $P_j(S) = 1$, that is, $j=p$.
\end{itemize}
\end{proposition}

\begin{proof}
If the first condition is satisfied, team $r$ can create a vacancy by making team $j$ the winner of tournament $(S,P)$, which is filled by team $r$ since $Q_j(T) \leq q-1$.

If the second condition is satisfied, team $i$ can create a vacancy by making team $j$ the winner of tournament $(S,P)$, which is filled by team $i$ since $Q_j(T) = q$.

If the third condition is satisfied, team $i$ can create a vacancy by ensuring that the winner of tournament $(S,P)$ is ranked $q$th in tournament $(T,Q)$, which is filled by team $i$.

The definition of the allocation rule $\mathcal{A}$ implies that no other team can receive the possible vacant slot, thus, incentive incompatibility arises only in the cases above.
\end{proof}

The first condition becomes relevant if a club (team $r$) that hopes to have the highest individual club coefficient in the champions path in the next season supports a club (team $j$) from a higher-ranked national association to win the UEFA Champions League in order to create a vacancy to be filled by itself. This is much less threatening than the third condition because tanking is excluded in the final knockout stage and the titleholder usually qualifies via its domestic championship, too. For example, the Ukrainian Shakhtar Donetsk (that had the highest coefficient in the champions path as we have seen in Section~\ref{Sec3}) played in a group with Antwerp, Barcelona, and Porto in the 2023/24 UEFA Champions League. Therefore, it would be interested in supporting Barcelona to win the Champions League at the expense of Antwerp since it was more likely that Barcelona would also qualify through its domestic championship.

The second condition says that a club (team $i$) that is ranked $(q+1)$th in its domestic league supports the club (team $j$) ranked directly above it from the same league to win the UEFA Champions League in order to create a vacancy to be filled by itself. However, this cannot be beneficial because (1) tanking is excluded in the final knockout stage and (2) teams from the same national association are not allowed to play against each other in the previous (group or league) stage. The latter observation highlights the role of draw constraints with respect to incentive compatibility, which is extensively studied by \citet{Csato2022a}.
Thus, neither hypothetical nor historical cases could be constructed to bridge the gap between theory and practice.

The example presented in Section~\ref{Sec3} exploits the third condition: $q=5$ implies the sixth-ranked Eintracht Frankfurt (team $i$) can benefit from swapping the fourth-ranked Borussia Dortmund (team $j$) and the fifth-ranked RB Leipzig (team $k$) if Dortmund plays in the Champions League and wins this competition.

\begin{proposition} \label{Prop2}
Tournament system $\mathbf{W} = \left[ (S,P), (T,Q), (U,R) \right]$ with the allocation rule $\mathcal{A}' (\mathbf{W})$ is incentive incompatible if and only if tournament $(S,P)$ is manipulable by the team $1 \leq r \leq \lvert U \rvert$ (for which $R_r(U) = 1$) such that the team $j$ in Definition~\ref{Def1} plays in both tournaments $(S,P)$ and $(T,Q)$ with $Q_j(T) \leq q$.
\end{proposition}

\begin{proof}
It follows from the proof of Proposition~\ref{Prop1}.
The second and third conditions in Proposition~\ref{Prop1} are eliminated because all cases $Q_p(T) \leq q$ are treated uniformly by the allocation rule $\mathcal{A}'$.
\end{proof}

\begin{proposition} \label{Prop3}
Tournament system $\mathbf{W} = \left[ (S,P), (T,Q), (U,R) \right]$ with the allocation rule $\mathcal{A}'' (\mathbf{W})$ is incentive incompatible if and only if tournament $(S,P)$ is manipulable by the team $1 \leq i \leq \lvert T \rvert$ for which $Q_i(T) = q+1$ such that the team $j$ in Definition~\ref{Def1} plays in both tournaments $(S,P)$ and $(T,Q)$ with $Q_j(T) = q$.
\end{proposition}

\begin{proof}
It follows from the proof of Proposition~\ref{Prop1}.
The first and third conditions in Proposition~\ref{Prop1} are eliminated because all cases $Q_p(T) \leq q$ are treated uniformly by the allocation rule $\mathcal{A}''$.
\end{proof}

The results above cannot be derived from the model of \citet{DagaevSonin2018} since the latter contains multiple round-robin and knockout tournaments.
However, if $(S,P)$ is a knockout tournament, then it cannot be manipulated by any team $1 \leq i \leq \lvert S \rvert$, and both allocation rules $\mathcal{A}'$ and $\mathcal{A}''$ guarantee incentive compatibility.

\section{Discussion} \label{Sec5}

Recent papers on fairness, fraud, and incentives in sport \citep{ScellesFrancoisValenti2024, VanwerschWillemConstandtHardyns2022} call for future research to accumulate more knowledge on this topic. Our paper examines the new format of the UEFA Champions League from the perspective of incentive compatibility, that is, whether the rules always provide the right incentives for the teams to perform.
The rules have clearly improved from the 2024/25 season. Now the clubs are seeded based on their individual club coefficients except for the guaranteed place of the titleholder in the first pot \citep[Article~16.01]{UEFA2024b}. Therefore, the regime used until the 2014/15 season has returned, and the seeding system satisfies incentive compatibility.

On the other hand, strategy-proofness is violated in the new format according to Section~\ref{Sec3}. The underlying reason is highlighted by the analysis of the mathematical model in Section~\ref{Sec4}, in particular, by Proposition~\ref{Prop1}: the European performance spots are allocated only \emph{after} the vacancies caused by the two titleholders are solved \citep[Article~3.07]{UEFA2024b}. Thus, a team playing in the same round-robin tournament (domestic league) as one of the (potential) titleholders is interested in \emph{not} creating vacancies in the UEFA Champions League. This can be achieved by ``helping'' other teams to be ranked over the (potential) titleholder. Such misaligned incentives may arise in the associations entitled to at least one spot in the league phase of the UEFA Champions League if they also earn one of the European performance spots. These are currently the 10 highest-ranked associations: in the 2024/25 season, England, Spain, Germany, Italy, France, the Netherlands, Portugal, Belgium, Scotland, and Austria.

Based on Section~\ref{Sec4}, incentive incompatibility is mitigated by changing the allocation rule according to the definition of either $\mathcal{A}'$ or $\mathcal{A}''$:
\begin{itemize}
\item
\emph{Rule change A}: Reverse the order of applying Articles~3.04 and 3.07 in the UEFA Champions League regulation \citep{UEFA2024b}. That is, first the two European performance spots are allocated \citep[Article~3.07]{UEFA2024b}, and then the possible vacancies created by the titleholders are addressed \citep[Article~3.04]{UEFA2024b}.
\item
\emph{Rule change B}: Modify the way of filling vacancies in \citet[Article~3.04]{UEFA2024b}. That is, any vacancy created by a titleholder is filled by the highest-ranked not already qualified club from its association.
\end{itemize}

If rule change A is chosen, Proposition~\ref{Prop2} remains valid. Hence, a tanking strategy can exist in the league phase of the UEFA Champions League for a team that qualifies for the champions path in the qualification of the next Champions League. For example, any team that is already eliminated in the league phase but can have the highest UEFA club coefficient in the champions path is interested in creating a vacant slot. This might be achieved by losing against a team that is expected to obtain a direct qualifying slot in its domestic league in order to help it win the Champions League. Consider the Swiss champion Young Boys in the 2024/25 UEFA Champions League. The club was eliminated in the league phase after six rounds and played against Celtic (Scotland) and Crvena Zvezda (Serbia) in the last two rounds. Since the Scottish champion directly qualifies for the league phase in contrast to the Serbian champion, Young Boys could have maximised its chances by losing against Celtic and exerting full effort against Crvena Zvezda.

On the other hand, if rule change B is chosen, Proposition~\ref{Prop3} should be considered. Hence, a tanking strategy can exist in the league phase of the UEFA Champions League for a team if the Champions League is won by a team from the same national association due to the tanking. However, this is impossible because two teams from the same country are not allowed to play against each other in the league phase \citep[Article~16.02]{UEFA2024b}, and a deliberate loss in the knockout stage cannot be beneficial.

Consequently, even though rule change A seems to be closer to the original intention of UEFA, it might still lead to incentive incompatibility, albeit with a lower probability. At least, the situation presented in Section~\ref{Sec3} will certainly be avoided.
In contrast, rule change B entirely avoids misaligned incentives and highlights the role of the association constraint in the league phase draw, which is desirable not only to maximise the number of international matches played but also to ensure strategy-proofness in the league phase. This is a good illustration of the idea raised by \citet{Csato2022a} that draw restrictions can be used to prohibit potentially unfair matches.
Nonetheless, rule change B favours the domestic league of the titleholder, hence, it leads to a somewhat higher concentration of the clubs playing in the UEFA Champions League by their associations because the highest-ranked countries are more likely to win both the UEFA Champions League and the UEFA Europa League.

To summarise, the main message of \citet{DagaevSonin2018} also applies to the model of Section~\ref{Sec4}: the vacant slots should be filled via the round-robin tournament as done by the allocation rule $\mathcal{A}''$. Therefore, UEFA could have guaranteed incentive compatibility by following the recommendation of \citet{DagaevSonin2018} despite the different settings.

In the example of Section~\ref{Sec3}, Rule change A means that the top five German teams enter the UEFA Champions League due to the European performance spot, and Borussia Dortmund would create a vacancy if it wins the 2023/24 UEFA Champions League, to be filled by the Ukrainian champion Shakhtar Donetsk.
On the other hand, Rule change B implies that the top five German teams enter the UEFA Champions League due to the European performance spot, and Borussia Dortmund would create a vacancy if it wins the 2023/24 UEFA Champions League, to be filled by the sixth-placed Eintracht Frankfurt.

If UEFA is not worried much about manipulation in a domestic league and is not willing to change its regulations, the national associations are also able to mitigate the problem to some extent. The current paper has considered incentive incompatibility as a binary concept. However, the probability of a situation that is vulnerable to tanking can be quantified via Monte Carlo situations, analogously to \citet{Csato2022a}, even if it would be difficult to implement.

Naturally, the set of matches to be played cannot be changed in a domestic championship without changing the tournament format, but the order of the matches is a decision variable. The scenario outlined in Section~\ref{Sec3} has become threatening especially since the match Eintracht Frankfurt vs.\ RB Leipzig was played in the last round. The set of top teams may be relatively accurately estimated at the beginning of the season, thus, the last rounds can be scheduled to avoid any matches between the strongest teams. Such a schedule might lead to a tanking opportunity at the end of the season with a (much) lower probability. But even if Eintracht Frankfurt had already played against RB Leipzig, it would be unfair if it turns out \emph{ex post} that a team has missed qualifying for the UEFA Champions League due to winning a match months before.

Furthermore, if strong teams do not play against each other in the last rounds, some strong teams may play against an opponent with nothing to compete for anymore, while other strong teams may play against opponents still in contention for spots in UEFA club competitions or to avoid relegation. The former matches can jeopardise fairness as discussed by \citet{CsatoMolontayPinter2024} and \citet{ScellesFrancoisValenti2024}.
Therefore, scheduling the domestic leagues accordingly is not a good solution for the problem identified here.
Consequently, it is the responsibility of UEFA to act if it is really committed to its announced principles.

\section{Conclusions and managerial implications} \label{Sec6}

This paper has revealed that the qualification system of the UEFA Champions League is incentive incompatible from the 2024/25 season. In particular, a game with misaligned incentives was narrowly avoided in the 2023/24 German Bundesliga.
We have used a formal mathematical model to uncover the underlying causes and to provide two alternative policies.
UEFA is strongly encouraged to adopt one of our recommendations in order to safeguard the integrity of its flagship tournament.

\begin{table}[t!]
\caption{Timeline of incentive incompatibility in UEFA club competitions \\ caused by the allocation of vacant slots}
\label{Table2}
\rowcolors{3}{}{gray!20}
\begin{tabularx}{\textwidth}{lL} \toprule
    Date/Period & Problem \\ \bottomrule
    2009/10--2014/15 & Incentive incompatibility of the UEFA Europa League entry \citep{DagaevSonin2018} \\
    May 2012 & The match Feyenoord vs.\ Heerenveen, where Heerenveen was better off by losing than by playing a draw \citep[Chapter~2.1]{Csato2021a} \\
    March 2013 & The first version of \citet{DagaevSonin2018} becomes available that verifies the incentive incompatibility of the UEFA Europa League entry and offers a general solution to the problem \\
    2015/16--2017/18 & Incentive incompatibility of the UEFA Champions League entry \citep{Csato2019d} \\
    2015/16--2023/24 & Incentive incompatibility of the UEFA Champions League seeding system \citep{Csato2020a} \\
    2024/25-- & Incentive incompatibility of the UEFA Champions League entry (this paper) \\ \bottomrule
	\end{tabularx}
\end{table}

The problem identified is especially worrying because it is rooted in an issue (the allocation of vacant slots) that was repeatedly responsible for incentive incompatibility in UEFA club competitions as summarised in Table~\ref{Table2}.
This offers an instructive lesson for decision-makers around the world: commitment to sporting merit requires not only media releases but also appropriately designed tournaments. Since rule changes can always affect fairness, a thorough assessment is necessary before they are implemented.

In order to achieve that fundamental aim, \citet{KendallLenten2017} suggest two alternative calls to action:
\begin{itemize}
\item
The governing bodies of major sports should invite academic representatives onto their committees, who would be tasked with identifying possible loopholes in proposed rule changes in consultation with the scientific community; or
\item
All proposed rule changes need to be posted in a public forum, where interested academics could comment, perhaps after running simulations.
\end{itemize}
Although UEFA works with academics both to improve strategic decision-making (UEFA Research Grant Programme, \url{https://uefaacademy.com/courses/rgp/}) and to support medical and anti-doping decision-making (UEFA Medical and Anti-Doping Research Grant Programme, \url{https://uefaacademy.com/courses/mrgp/}), these initiatives seem to be inadequate to address fairness issues in UEFA club competitions as illustrated by the examples in Table~\ref{Table2}. Hence, a potential solution can be to invite members of the academic community to study every planned format and rule changes in European football.

\section*{Acknowledgements}
\addcontentsline{toc}{section}{Acknowledgements}
\noindent
We are grateful to \emph{Dries Goossens} and \emph{David van Bulck}, who have organised the third \href{https://robinxval.ugent.be/FairnessInSports/}{Fairness in Sports workshop} in Ghent on 11 June 2024. \\
\emph{Phil Scarf}, two reviewers, and one anonymous colleague provided valuable comments and suggestions on earlier drafts. \\
We are indebted to the \href{https://en.wikipedia.org/wiki/Wikipedia_community}{Wikipedia community} for summarising important details of the sports competitions discussed in the paper. \\
The research was supported by the National Research, Development and Innovation Office under Grant FK 145838, and the J\'anos Bolyai Research Scholarship of the Hungarian Academy of Sciences.

\bibliographystyle{apalike}
\bibliography{All_references}

\end{document}